\documentclass{article}

\usepackage[utf8]{inputenc}		
\usepackage{authblk}

\usepackage{amsmath}
\usepackage{amsthm}
\usepackage{amsxtra}
\usepackage{mathpazo}
\usepackage{MnSymbol}
\usepackage{marvosym}

\usepackage{enumerate}
\usepackage{braket}
\usepackage{tikz} 
\usepackage{url}
\usepackage{hyperref}

\DeclareMathAlphabet{\mathsc}{OT1}{cmr}{m}{sc}
\allowdisplaybreaks

\newcommand{\lo}[1]{\raisebox{-0.1ex}{$#1$}\,}
\newcommand{\loo}[1]{\raisebox{-0.2ex}{$#1$}\,}
\newcommand{\Lo}[1]{\raisebox{-0.3ex}{$#1$}\,}

\newcommand{\R}{\mathbb R}

\newcommand{\N}{\mathbb N} 
\newcommand{\Z}{\mathbb Z}
\newcommand{\Q}{\mathbb Q}

\newcommand{\eps}{\varepsilon}
\newcommand{\abs}[1]{\lvert #1 \rvert}

\newcommand{\norm}[1]{\lVert #1 \rVert}

\newcommand{\D}[1]{\mathrm{d}#1}
\newcommand{\I}{\mathrm{i}}
\newcommand{\e}{\mathrm{e}}
\newcommand{\Langle}{\left\langle}
\newcommand{\Rangle}{\right\rangle}
\newcommand{\Cb}[1]{C_\mathrm{b}\!\left(#1\right)}

\newcommand{\mc}[1]{\mathcal{#1}}

\DeclareMathOperator{\tr}{tr}

\theoremstyle{definition}
\newtheorem{defn}{Definition}

\theoremstyle{plain}
\newtheorem{prop}[defn]{Proposition}
\newtheorem{thm}[defn]{Theorem}
\newtheorem{coro}[defn]{Corollary}

\theoremstyle{remark}
\newtheorem{ex}[defn]{Example}
\newtheorem{rk}[defn]{Remark}

\title{Cylinder Measures \\in Renormalization Theory}
\author{Rodrigo Vargas Le-Bert\footnote{
\Letter {\tt rodrigovargas@uach.cl}. 
} }
\affil{
Instituto de Ciencias Físicas y Matemáticas\\ Universidad Austral de Chile\\
Casilla 567, Valdivia, Chile}
\date{\today}

\begin{document}

\maketitle

\begin{abstract}
The standard approach to renormalization relies, technically, on the asymptotic perturbation of Gaussian measures embodied in Feynman diagram theory. From a mathematical standpoint this is not good enough, because thereby solving the renormalization problem does not immediately traduce into having a rigorous construction of the corresponding theory. Here, we start developing an approach to renormalization based on cylinder measures.
After explaining how renormalization can be mathematically better understood in terms of them, we argue that a renormalization problem can, under certain hypothesis, be reduced to that of the corresponding strongly coupled theory, and obtain a family of solutions for the case of scalar bosons whose interaction Lagrangian does not contain derivatives. As a further application, we produce an explicit, formal expression for the cylinder measure of a local field with effectively quartic interaction at any given, fixed scale, in arbitrary dimension.
\end{abstract}

\setcounter{tocdepth}{2}
\tableofcontents

\section{Introduction}

After Nelson's work on the reconstruction theorem of quantum fields from Markov fields~\cite{nelson1973construction}, which was turned into an axiomatic approach to Quantum Field Theory by Osterwalder and Schrader~\cite{osterwalder1973axioms} and developed into a full correspondence between OS-positive stochastic processes and stochastically positive quantum systems by Klein and Landau~\cite{klein1981stochastic}, it is well understood how constructing quantum fields is tantamount to constructing certain measures on spaces of classical fields. However, the technical and conceptual problems with measure theory on path spaces, together with the availability of strong guidance coming from physical intuition, have made the subject grow mostly dominated by an asymptotic perturbation theory of Gaussian measures. Here, we start developing an approach to the construction of quantum fields using better tools from an analyst perspective, motivated by the conviction that 
the task has an unavoidable essential difficulty of analytical nature.
In doing so, we believe that we have also contributed to the clarification of the renormalization problem, which, at a mathematical level of precision, does not seem to have reached a definitive formulation yet.

The difficulty with measure theory on path spaces that we deal with here is of a calculational nature. The problem is that producing explicit, well-defined expressions for infinite dimensional measures which are not simple products can be difficult, and it would seem that we know how to treat essentially one  example: the Gaussian. The mathematical way of providing coordinate expressions for infinite dimensional measures is the theory of cylinder measures, which are compatible collections of measures on the finite dimensional quotients of the space of interest. This goes along well with a more physical perspective for, as we will see, an effective theory can be thought of as a measure on a particular kind of quotient of the space of fields, and a compatible collection of effective theories defines, thus, a cylinder measure. Renormalization, therefore, can be seen as the problem of constructing certain cylinder measures for which we have formal, ill-defined expressions coming from physical considerations.


The contents of this paper are as follows. We start by recalling the theory of cylinder measures, introducing a slight generalization which will enable us to treat some non-trivial cases of physical interest. Then, we study the problem of explicitly constructing a cylinder measure, starting from the kind of formal expression typically used to encode the physical requirements. 
An obvious, first approach is to focus on the perturbation of the free measure obtained for small values of the coupling constant. That point of view gives us some interesting insights on the nature of renormalization, but we do not stop there because no clear simplification of the cylinder measure construction problem is therefore achieved. Consequently, we devise a second method based on a perturbation series presuposing the knowledge of a solution of the renormalization flow in the strong coupling regime, and then we construct a family of exact solutions to that flow for the case of scalar bosons whose interaction Lagrangian does not contain derivatives. We close by producing an explicit, formal expression for the cylinder measure of a local field with effectively quartic interaction at any given, fixed scale, in arbitrary dimension.


\section{Cylinder measures and effective theories}

The construction of a measure on an infinite dimensional space usually proceeds, as in Gross or Minlos theories~\cite{gross1967abstract, minlos1959generalized}, by first constructing a finitely additive measure and then radonifying it on a suitable enlargement of the space in question (in Kolmogorov's approach one starts with a space so large that the measure is already $\sigma$-additive, and then the problem is to find a suitable smaller space of full measure).
There is a finite dimensional example that can be very enlightening in this regard: think of the ``measure'', defined on the  lattice of closed subsets of $\Q$, by $\mu([a,b]) = b-a$. Of course, $\mu$ is not $\sigma$-additive, but it should not be blamed for that: the problem lies in the ambient space. This points to a limitation of Carathéodory's stretegy for the construction of measures, which is: start with a measure defined on an algebra, show that it has the desired regularity there (one typically asks for $\sigma$-additivity, but $\tau$-additive measures can also be treated, see~\cite{konig2000measure, konig2012measure}), complete the algebra in accordance with the regularity, and finally extend the measure. Instead, as in Gross or Minlos theories, one can complete the \emph{ambient} space in such a way that the measure \emph{acquires} the desired regularity. In other words, one should work first in coordinates (a level at which we cannot ask for more than finite additivity) and then regularize in a suitable completion of the ambient space. The finitely additive measure one starts with is typically a cylinder measure. 


\subsection{Cylinder measures}

Cylinder measures are traditionally defined as compatible collections of measures on the set of all finite dimensional quotients of a locally convex vector space---see, for instance,~\cite{maurey1972probabilites, badrikian1969mesures}. This is insatisfactory for two reasons: first, it innecessarily leaves out non-linear spaces, hiding the purely coordinate-system nature of the notion; second, and more important in the context of this work, in demanding for a measure on \emph{every} finite dimensional quotient it becomes impractical, leaving us with essentially one example: Gaussian measures. Thus, we adopt a different definition, and that will be fundamental to our treatment of renormalization theory.

Let $\mc P$ be a directed set, $\set{X_P|P\in\mc P}$ a projective system of topological spaces with projective limit $\overline X$ and canonical projections $\pi_P:\overline X\rightarrow X_P$\loo, and $X$ a subspace of $\overline X$ which is \emph{full,} in the sense that $\pi_P(X)=X_P$\loo. By a harmless abuse of notation, we will usually write $P$ instead of $\pi_P$\loo.
It will also be convenient to write $P$ for the projection $X\rightarrow X_P$\loo, and even for the projection $X_Q\rightarrow X_P$ when $Q\succcurlyeq P$ is understood from the context. 
We think of $\set{P:X\rightarrow X_P|P\in\mc P}$ as a coordinate system on $X$. 
Now, given $P\in\mc P$, consider the algebra
\[
\Cb{X;P} = \Set{\vphantom{\hat A} f\in \Cb{X} | f \text{ factors through } P:X\rightarrow X_P },
\]
where $\Cb{X}$ stands for the space of bounded, continuous functions on $X$.
If $P\preccurlyeq Q$, there is a natural inclusion $
\Cb{X;P}\hookrightarrow \Cb{X;Q}$. The resulting directed system has an algebraic injective limit 
\[
\Cb{X;\mc P} = \injlim \Set{\vphantom{\hat A} \Cb{X;P} | P\in\mc P }.
\]
The elements of $\Cb{X;\mc P}$ are called \emph{cylinder functions} (for the coordinate system in use). 
Some important examples of coordinate systems follow.

\begin{ex}
Let $X$ be a Tychonoff space and $\Set{e_i|i\in\mc I}\subseteq \Cb{X}$ a separating family of continuous functions. Given a finite subset $I\subseteq\mc I$, consider the equivalence relation
\[
x\sim y \Leftrightarrow (\forall i\in I)\ e_i(x)=e_i(y)
\]
and let $P_I:X\rightarrow X_I$ be the corresponding quotient. One has that $X$ is a full subspace of $\projlim X_I$\loo. 
Now, for $I=\{i_1\lo,\dots,i_n\}$, we define 
\[
e_I:X\rightarrow\R^n,\quad e_I(x) = \left(e_{i_1}(x),\dots e_{i_n}(x)\right)
\]
and then $f\in\Cb{X}$ is a cylinder function if, and only if, it factors through one of the $e_I$'s.
\end{ex}

\begin{ex}
Consider a path space $X_I=C(I,X)$, with $I\subseteq\R$. Given a finite number of time instants $t_1\lo,\dots,t_n\in I$, we have the projection
\[
x\in X_I\mapsto \left( x_{t_1}\Lo,\dots, x_{t_n} \right) \in X^n.
\]
The resulting coordinate system on $X_I$ has cylinder functions 
\[
x\mapsto f\left( x_{t_1}\Lo,\dots x_{t_n} \right),\quad f:X^n\rightarrow\R.
\]
The classical Kolmogorov consistency theorem is about the construction of path space measures in this coordinate system.
\end{ex}

\begin{ex}
Let $X$ be a Banach space and $\mc P\subseteq B(X)$ a directed family of projections converging strongly to $1\in B(X)$ (so that $X$ has the metric approximation property). The importance of this convergence hypothesis will be shortly seen. This is the example that corresponds more closely to the situation studied in standard cylinder measure theory. A particular case is that of a separable Hilbert space with orthonormal basis $\set{e_i}$ and projections $P_n = \sum_{i\leq n}e_ie_i^*$\loo, where $e_i^*(x) = \langle e_i\lo,x\rangle$. Cylinder functions, then, are those which depend only on a finite number of coordinates.
But the projections in $\mc P$ might well not be finite-dimensional, so that our framework also generalizes that of $G$-cylinder measure theory, in Maurey's terminology~\cite{maurey1974rappels}.
\end{ex}

\begin{ex} \label{deltas as coordinates}
In this example, which is perhaps the most important in the context of field theory, we provide a rigorous version of the notion of 
\[
\Set{\delta_s=\delta(\cdot-s)|s\in S},
\]
where $\delta$ is Dirac's delta, as a ``basis'' for a space of fields $x:S\rightarrow\R$. Given a measure space $(S,\D s)$, take a system $\set{p_i | i=1\dots n}$ of projections of the von Neumann algebra $L^\infty(S)$ which is orthogonal and complete, in the sense that $p_ip_j=0$ and $\sum p_i=1$. 
To $\{p_i\}$ we associate the conditional expectation 
\[
P:L^\infty(S)\rightarrow X_P\loo,\quad
P = \sum_i p_ip_i^*\loo,\ p^*(x) = \int_S \bar px,
\]
where $\bar p = p/\abs p$ and $\abs p = \int_X p$. 
Now, let $\set{q_{ij} | i=1\dots n,\,j=1\dots m}$ be a refinement of $\{p_i\}$, i.e.\ another complete system of orthogonal projections such that $p_i=\sum_{j} q_{ij}$, with associated conditional expectation $Q:L^\infty(S)\rightarrow X_Q$\loo. Since $\{q_{ij}\}$ is a refinement of $\{p_i\}$, we have a projection (conditional expectation) $X_Q\rightarrow X_P$\loo. Given a directed family $\mc P$ of such systems of orthogonal projections we get a projective system $\{X_P\}$ and, if the family generates $L^\infty(S)$, then any good $X\subseteq L^\infty(S)$ will become a full subspace of $\projlim X_P$\loo. 
\end{ex}

\begin{defn}
Let $X$ be a Tychonoff space equipped with a coordinate system $\Set{P:X\rightarrow X_P|P\in\mc P}$.
A \emph{cylinder measure} on $X$ is a family of Radon measures $\set{\mu_P\text{ on } X_P}$ which is \emph{compatible,} in the sense that
\[
P_*\mu_Q = \mu_P\lo,\quad\text{ for all } Q\succcurlyeq P.
\]
We can also adopt a dual point of view and define a cylinder measure as a compatible family of positive linear functionals $\set{\rho_P:\Cb{X_P}\rightarrow\R}$---or, in other words, a positive linear functional on the injective limit $\Cb{X;\mc P}$.
\end{defn}

\subsection{Effective field theories}

Constructing a field theory consists, essentially, in making sense of certain formal, ill-defined probability density functional on certain path space $X$ (which, itself, is not given in full detail). To begin with, one has a cylinder Gaussian measure $\mu(\D x)$ together with a potential, or interaction Lagrangian $V:X\rightarrow \R$, and one has to give a meaning to the formal expectation values
\begin{equation} \label{cylinder expectation value}
\frac{\int f(x)\e^{-V(x)}\mu(\D x)}{\int \e^{-V(x)}\mu(\D x)}
\end{equation}
where $f:X\rightarrow\R$ is a cylinder function. The problem is that $V$ is typically not a cylinder function; an obvious (failure bound) solution attempt is to introduce suitable finite dimensional projections $P_n:X\rightarrow X$ converging to the indentity operator, replace $V$ by the cylinder function $V_n(x) = V(P_nx)$ (which amounts to regularizing $V$) and hope that the resulting quotients converge. Now, the limit of $\int \e^{-V_n(x)}\mu(\D x)$ as $n\rightarrow\infty$ is either strictly positive, or zero (we remark that continuous, non-zero functions can integrate zero with respect to a finitely additive measure). In the second case, which should be expected to be the rule rather than the exception, one should also have $\int f(x)\e^{-V(x)}\mu(\D x) = 0$ for a big class of functions $f:X\rightarrow\R$, so that a priori the quotient \eqref{cylinder expectation value} could still make sense. However, there is evidence that~\eqref{cylinder expectation value} cannot define a measure unless $\int\e^{-V(x)}\mu(\D x)>0$ (we will provide some later); therefore, the family $\{V_n\}$ of effective Lagrangians will typically have to be modified---and, indeed, the resulting cylinder measure does not even need to arise from a Lagrangian as in equation~\eqref{cylinder expectation value}, for there is no need of convergence for the modified $V_n$'s. This is one way of seeing the need for renormalization and provides an understanding of what kind of objects should count as valid generalized Lagrangians,
%
but from this perspective we have not yet reached a precise formulation of the renormalization problem.
The strategy leading to the understanding that we have just summarized was to study first, on their own right, the cylinder measures that could arise as its solutions.  

Let us get into matter.
Let $S$ be some space-time 
and $X$ a suitable space of fields on $S$.
There are two natural choices for the coordinate system 
\[
\Set{P:X\rightarrow X_P | P\in\mc P} 
\]
which we shall explore shortly. We adopt the point of view encoded in the following definition.
\begin{defn}
An \emph{effective theory} is a measure $\mu_P$ on $X_P$\loo. Given $Q\preccurlyeq P$, we can push $\mu_P$ forward to get a measure on $X_Q$\loo. Thus, a cylinder measure is a compatible family of effective theories.
\end{defn}

\subsection{Physical space coordinates}

Here, we consider the field probability measure in
the coordinate system given by the projections of the von Neumann algebra $L^\infty(S)$---see Example~\ref{deltas as coordinates}.
In order to get a better feeling for the difficulty of constructing the cylinder measure, 
let us attempt a direct calculation of one renormalization step, which is: given a suitable measure on $X_Q$\loo, to  push it forward via $P:X_Q\rightarrow X_P$\loo. 
The measure $\nu_Q$ on $X_Q$ to be pushed-forward should effectively come from adding a local potential to some free process  $\mu_Q$\loo. We recall that if $\mu_Q$ is Gaussian, then its density $\lambda_Q$ reads, in the coordinates $x_{ij} = q_{ij}^*(x_Q)$,
\[
\lambda_Q\bigl((x_{ij})\bigr) = \int \frac{\e^{\I\sum\xi_{ij} x_{ij}}\prod\D\xi_{ij}}{(2\pi)^{nm}}  \exp\left(-\frac{1}{2}\sum\xi_{ij}\xi_{kl}
\Langle\bar q_{ij}\lo,\Gamma\bar q_{kl}\Rangle\right).
\]
The covariance operator $\Gamma$ is, in the case of a free boson field, the Green function $(-\Delta+m^2)^{-1}$ with appropriate boundary conditions.
There is a question, now, as to what a local potential is, given that space has been effectively discretized. If we admit potentials depending on spatial derivatives, then discretizing space inevitably makes all potentials ``local'', even if perhaps dependent on unreasonably high-order finite differences. So, let us just consider potentials of the form
\begin{equation} \label{effective local potential}
V_Q(x_Q) = \sum_{i,j} \abs{q_{ij}}f_Q(x_{ij}),\quad f_Q:\R\rightarrow\R
\end{equation}
which, besides being (strongly) local, are spatially homogeneous. 

The next step is to choose coordinates on $X_Q$ which are better adapted than $x_{ij} = q_{ij}^*(x_Q)$ to integrate out the undesired degrees of freedom. This amounts to choosing a basis of $X_Q$ compatible with the decomposition $X_Q = X_P\oplus\ker P$, and it is convenient to do so by extending the basis $\{p_i\}$ of $X_P$\loo. Let us write $\{p_i\}\cup\set{p_{ij}|j=1\dots m-1}$ for the extension, and let 
\[
x_i=p_i^*(x_Q),\quad x'_{ij}=p_{ij}^*(x_Q), 
\]
where $\{p_i^*,p_{ij}^*\}$ is the dual basis. Thus, we have to integrate out the primed variables. We warn the reader that we will find it convenient to use the following, potentially confusing notation: we will simply write $x$ for $(x_i)$, not to be mistaken for $(x_{ij})$, and $x'$ for $(x'_{ij})$.

There is a trade-off in the election of the $p_{ij}$'s: one has to aim at simplifying the change of coordinates of either the space variables, $x_{ij} = x_{ij}(x,x')$, or their duals 
\[
\xi_{ij} = \xi_Q(q_{ij}),\quad \xi = \bigl(\xi_Q(p_i)\bigr),\quad \xi' = \bigl(\xi_Q(p_{ij})\bigr), 
\]
where $\xi_Q\in X_Q^*$ (note that we would like both to be simple, because locality is defined in terms of space variables, whereas the free measure is given by its characteristic function). Simplification of the former occurs if
we define, for instance,
\(
p_{ij} = q_{ij} - \frac{\abs{q_{ij}}}{\abs{q_{im}}}q_{im}\loo,
\)
so that 
\[
x_{ij} = \left\{\begin{aligned}
&x_i+x'_{ij}& &j<m,\\
&x_i-\sum_{j'\neq 1} \frac{\abs{q_{ij'}}}{\abs{q_{im}}}x'_{ij'}& &j=m,
\end{aligned}\right.
\]
which can be further simplified if we introduce extra variables $x'_{im}$ together with the restrictions $\sum_j \abs{q_{ij}}x'_{ij}=0$. 
Something similar is obtained for the dual variables if we use, instead,
\(
p_{ij} = \frac{\abs{q_{ij}}}{\abs{p_i}}p_i - q_{ij}\loo,
\)
for then
\[
\xi_{ij} = \frac{\abs{q_{ij}}}{\abs{p_i}}\xi_i - \xi'_{ij}\Lo,
\]
where, as before, we have introduced the extra variables $\xi'_{im}$ together with the restrictions $\sum_{j}\xi'_{ij}=0$. Just note that, in any case, the density $\lambda_Q(x,x')$ will not decouple, in the sense of being a tensor product, i.e.\ a product of functions of $x$ and $x'$ separately: such decoupling occurs in momentum space coordinates.

Whatever coordinates we choose to use, the integral behind one renormalization step will have the form
\begin{equation} \label{renormalization integral}
\nu_P(\D x_P) = \biggl(\int\D x' \e^{-V_Q(x,x')}\lambda_Q(x,x') \biggr)\D x.
\end{equation}
\begin{rk}
Regarding the notion of renormalization semigroup, observe that
there is no reason at all to expect that there exists a function $f_P:\R\rightarrow\R$ such that
\[ 
\int\D x' \e^{-\sum_{i,j} \abs{q_{ij}}f_Q(x_{ij})}\lambda_Q(x,x') = \e^{-\sum_i \abs{p_i} f_P(x_i)}\lambda_P(x),
\] 
even if we allow for a modification of the free part (keeping it Gaussian). Otherwise said, a local interaction becomes non-local, or at least not as local as in~\eqref{effective local potential}, after one (exact) renormalization step. Now, of course, for an effective theory the standard of locality should be relaxed, up to something of the order of the scale, and it does seem plausible that asymptotically safe theories, for instance, could be defined by families $\{f_P\}$ such that
\[
\int\D x' \e^{-\sum_{i,j} \abs{q_{ij}}f_Q(x_{ij})}\lambda_Q(x,x') = \e^{-\sum_i \abs{p_i} f_P(x_i) + O(\abs{P})}\lambda_P(x),\quad \abs{P} = \max\abs{p_i}.
\]
However, in general one should still consider the possibility that effectively non-local theories arise from local ones. 
\end{rk}

The integral~\eqref{renormalization integral} is quite intractable, because all of its variables are coupled. Indeed, the (strong) locality notion that we are using is that $\e^{-V_Q}$ decouples when written in physical space variables $\{x_{ij}\}$; therefore, the interaction potential produces a coupling which is internal to each of the spatial regions encoded by the $p_i$'s. However, the free part couples all of the $x_i$'s. 
The resulting integral is irreducibly high-dimensional: different momentum eigenstates become coupled by a local potential---or, from our current coordinate point of view, different spatial sites become coupled by the free dynamics.

\subsection{Momentum space coordinates}

Physical space coordinates have two clear disadvantages: firstly, the free part couples all sites; secondly, one has to make a change of variables in order to perform one renormalization step. Both problems are solved in momentum space coordinates, at the expense that now the interaction potential couples all modes. Depending on the problem under study, this can be a good deal.

We will only consider momentum space coordinates when $S=\mathbb T^{d+1}$---that is, for infrared cutoff theories at non-zero temperature. Then, we 
let $X_n\subseteq L^2(S)$ be the space of trigonometric polynomials of degree $n$, with $P_n:L^2(S)\rightarrow X_n$ given by the truncated Fourier series. Again, the projective system $\{X_n\}$ defines a coordinate system on any reasonable space of fields $X\subseteq L^2(S)$. As in the case of physical space coordinates, we can express this in terms of the projections of a von Neumann algebra, namely $\ell^\infty(\Z^{d+1})$, the difference being that this time we have minimal projections at our disposal. Note that if one space-time direction ceases to be compact, then we lose the corresponding minimal projections but the momentum coordinate system can still be defined, as we did with the physical space coordinates.

The existence of minimal projections simplifies the construction of cylinder measures. Indeed, let $\Set{e_k | k\in\Z^{d+1}}$ be the  basis of $L^2(S)$ given by 
\[
e_k(s) = \prod_{i=1}^{d+1} e_{k_i}(s_i),\quad e_{k_i}(s_i) = 
\left\{\begin{aligned}
&\cos(k_is_i)& &k_i\geq 0, \\
&\sin(-k_is_i)& &k_i<0,
\end{aligned}\right.
\] 
and write $x= \sum \hat x_ke_k$\loo. There is a whole family of cylinder measures which are trivial to define, for they are simply formal products, i.e.\ measures of the form
\[
\prod_{k\in\Z^{d+1}} u_k(\hat x_k)\D \hat x_k\loo,\quad \int_\R u_k(\hat x_k)\D \hat x_k = 1.
\]
The free measure is one of them: it is given by
\[
\mu(\D x) = \prod_{k\in\Z^{d+1}} C_k\e^{-\frac{1}{2}\lambda_k\hat x_k^2}\D\hat x_k\loo,\quad \lambda_k = m_0^2+2\pi k^2.
\]
Now, if one decides to attempt renormalization in this coordinate system, the first problem is: which family of interaction potentials to work with? What would be a good notion of effective locality? A reasonable choice is to declare a potential on $X_n$ to be local if it can be put in the form
\[
V(x)=\frac{1}{n^{d+1}}\sum_{\ell\in\Z^{d+1}} f(x_\ell),\quad x_\ell=\sum_{k\in\Z^{d+1}} \hat x_ke_k(s_\ell),\quad s_\ell=2\pi\ell/n^{d+1}.
\]
Then, it might be convenient to take renormalization steps $X_{2^{n+1}}\rightarrow X_{2^n}$\loo, so as to take advantage of fast Fourier transform formulas in computing the effective potentials. Nevertheless, further simplifications seem necessary in order to make complete, explicit calculations.

\section{Renormalization of boson fields}

As we have seen, a direct understanding of the renormalization flow would involve working with an infinite number of  high-dimensional integrals, and some simplification is called for. 
Here, we propose two perturbative strategies which are better suited for an analytic treatment than the usual asymptotic expansion approach based on Feynman diagrams. Our developments are formal, for at this stage we are not yet in position to provide a general theory. 

\subsection{Cylinder perturbations}

Consider the following problem: to find a family of functions $\{V_P:X_P\rightarrow\R\}$ such that, in a suitable sense, 
\begin{equation}\label{cylinder perturbation}
P_*\left(\e^{-\eps V_Q}\mu_Q\right) = \e^{-\eps V_P}\mu_P + O(\eps^2)
\end{equation}
whenever $P\preccurlyeq Q$. We will call such a family a \emph{cylinder perturbation.} 

In order to approach this problem, choose a splitting $X_P\hookrightarrow X_Q$ and coordinates $x,x'$ on $X_P\lo,\ker P$, respectively. Then, let $\lambda_P(x)$ and $\lambda_Q(x,x')$ be functions such that
\[
\mu_P(\D x_P) = \lambda_P(x)\D x,\quad \mu_Q(\D x_Q) = \lambda_Q(x,x')\D x\D x'.
\]
Now, the compatibility between $\mu_P$ and $\mu_Q$ reads $\lambda_P(x) = \int\D x'\, \lambda_Q(x,x')$ and equation~\eqref{cylinder perturbation} becomes
\[
\int\D x\int\D x'\, f(x)\e^{-\eps V_Q(x,x')}\lambda_Q(x,x') = \int\D x\, f(x)\e^{-\eps V_P(x)}\lambda_P(x) + O(\eps^2),
\]
for all $f\in C(X_P)$. 
Expanding the exponential as a power series in $\eps$ we find that 
\[ 
\int\D x'\, \e^{-\eps V_Q(x,x')}\lambda_Q(x,x') = \lambda_P(x) - \eps\int\D x'\, V_Q(x,x')\lambda_Q(x,x') + O(\eps^2)
\] 
and, therefore, equation~\eqref{cylinder perturbation} can only hold if
\begin{equation} \label{cylinder perturbation compatibility}
V_P(x) = \frac{1}{\lambda_P(x)}\int\D x'\, V_Q(x,x')\lambda_Q(x,x'),
\end{equation}
suggesting to work in momentum space coordinates, for which $\lambda_Q(x,x')$ factors as a function of $x$ times a function of $x'$.

\begin{rk}
In order for this to provide a complete approach to the perturbative construction of path space cylinder measures, a technical result is missing: the existence, given a cylinder perturbation $\{V_P\}$ of $\{\mu_P\}$ and a small $\eps>0$, of a cylinder measure $\{\nu_P\}$ with $\nu_P = \e^{-\eps V_P}\mu_P + O(\eps^2)$.
\end{rk}

Equation~\eqref{cylinder perturbation compatibility}  is a workable condition for an interaction Lagrangian to define a cylinder perturbation, and thus provides a good starting point in delineating the renormalization problem.
Indeed, working in momentum space coordinates, the family $\bigl\{V_{P_n}\bigr\}$\lo, 
\[
V_{P_n}(x) = \int_{(1-P_n)X} V(x+x') (1-P_n)_*\mu(\D x'),
\]
would be the cylinder perturbation corresponding to the measure $\e^{-\eps V(x)}\mu(\D x)$. In particular, one must have $\int V(x)\mu(\D x) < \infty$---a condition showing, for instance, that a field with pure quartic self-interaction cannot exist, for the free measure does not radonify on $L^4(S)$ (see the appendix) and therefore
\[ 
\int_{X}\norm x_4^4\mu(\D x) = \infty.
\] 
Recall, indeed, that the $\phi^4$ field has been shown to exist in dimension $1+1$ only once the interaction term is put in Wick order~\cite{glimm1968lambda, glimm1968boson2, glimm1970lambdaII, glimm1970lambdaIII}, and in dimension $2+1$ further renormalization terms are needed~\cite{glimm1968boson3, glimm1973positivity}.

\subsection{Power series on the cutoff scale}

Although equation~\eqref{cylinder perturbation compatibility} 
has given us some insight on the renormalization problem,
it does not seem to lead to a tractable approach to the construction of cylinder measures. Now, instead of considering small values of the coupling constant, suppose that we are given the densities of two cylinder measures and have an interest in the measure with product density. This makes sense because, as we shall shortly see, the interaction part becomes more tractable if we ignore the free part, and then the problem is to recombine the two.

Given a coordinate system $\Set{P:X\rightarrow X_P | P\in\mc P}$, assume that there exists a countable, cofinal set $\bigl\{P^{(n)}\bigr\}\subseteq\mc P$ and write $X^{(n)} = P^{(n)}X$. Choose splittings $X^{(n)}\hookrightarrow X^{(n+1)}$ and coordinates $x_1$ on $X^{(1)}$, $x_k$ on $\ker \bigl( X^{(k)}\rightarrow X^{(k-1)}\bigr)$ for $k>1$. Suppose that we are given two cylinder measures $\bigl\{\mu^{(n)}\bigr\}$ and $\bigl\{\nu^{(n)}\bigr\}$. We will treat them asymetrically, thinking of $\mu$ as a reference and $\nu$ as a perturbation. So, write
\[
\mu^{(n)}\bigl(\D x^{(n)}\bigr) = \mu_1(\D x_1)\mu_2(x_1;\D x_2)\cdots\mu_n(x_1\lo,\dots,x_{n-1};\D x_n)
\]
and suppose that $\nu^{(n)}$ has density $f^{(n)}$ in the $x_k$ coordinates.
Now, using the convention $f^{(0)}=0$, define $\Delta f^{(n)} = f^{(n)}-f^{(n-1)}$, so that:
\begin{enumerate}
	\item $f^{(n)} = \sum_{k=1}^n\Delta f^{(k)}$.
	\item $\int\D x_n\, \Delta f^{(n)} = 0$ whenever $n>1$.
\end{enumerate}
We are interested in the measure which we might formally write as $f(x)\mu(\D x)$. For $n$ large, this measure should be well approximated on $X^{(n)}$ by
\(
f^{(n)} 
\mu^{(n)},
\)
and we want to compute the corrections to this approximation coming from the cutoff variables. This might be done as follows: if the corrected measure is $\tilde f^{(n)}\mu^{(n)}$, then
\begin{align*}
\tilde f^{(n)} &= \lim_{k\rightarrow\infty} \int\mu_{n+1}\cdots\int\mu_{n+k}\, \left( f^{(n)}+\Delta f^{(n+1)}+\cdots+\Delta f^{(n+k)} \right) \\
	&= f^{(n)} + \int\mu_{n+1}\, \Delta f^{(n+1)} + \int\mu_{n+1}\int\mu_{n+2}\, \Delta f^{(n+2)} +\cdots
\end{align*}
Here, we are using the notation $\int\mu_k f = \int\mu_k(\D x_k)f$, where $f$ might be a function of variables other than $x_k$\loo, in which case partial integration is meant. As a direct calculation immediately shows, in this way we obtain a formally compatible family of (formal) effective measures. In concrete applications, convergence and compatibility would have to be proved.

\subsection{Renormalization of the interaction potential}

Let us see how a local interaction part can be treated in physical space coordinates if we ignore the free part. Since the resulting coupling occurs only internally to each $x$ site,~\eqref{renormalization integral} factors into $n$ low-dimensional integrals. Thus, we can drop the $i$ indices in the projection systems $\{p_i\}$ and $\{q_{ij}\}$. For the sake of definiteness, assume that our $(d+1)$-dimensional space-time $S$ is divided into hyper-cubed regions, with $p$ corresponding to one of them, and that this region is, in turn, subdivided into $m=2^{d+1}$ hyper-cubes, with projections $\{q_j\}$. Thus, $\abs{q_j} = \abs{p}/m$.
We will use the variables
\[
\bar x = p^*(x_Q),\quad x_j = q_j^*(x_Q),\quad x'_j = \left( q_j^*- q_m^* \right)(x_Q)
\]
so that 
\[
x_j = \left\{\begin{aligned}
&\bar x + x'_j& &j<m,\\
&\bar x - \sum_{j=1}^{m-1}x'_j& &j=m.
\end{aligned}\right.
\]
Now, writing $\D x=\D x_1\cdots\D x_m$ and $\D x' = \D x'_1\cdots \D x'_{m-1}$ (exterior products are meant), one has that
\begin{align*}
\D x &= \bigl(\D\bar x+\D x'_1\bigr)\cdots\bigl(\D\bar x+\D x'_{m-1}\bigr)\bigl(\D\bar x - (\D x'_1+\cdots+\D x'_{m-1})\bigr) \\
	&= \sum_{k=1}^{m-1} \D x'_1\cdots \D x'_{k-1}\D\bar x\D x'_{k+1}\cdots \D x'_{m-1}(-\D x'_k) + \D x'\D \bar x \\
	&= m\D x'\D\bar x.
\end{align*}
Thus, writing $\mu_P(\D x) = \prod_i u_P(x_i)\D x_i$\lo, 
one obtains
\[ 
u_P(\bar x) = m\int\D x'\, u_Q\bigl(\bar x+x'_1\bigr)\cdots u_Q\bigl(\bar x+x'_{m-1}\bigr)u_Q\bigl(\bar x - (x'_1+\cdots+ x'_{m-1})\bigr),
\] 
which can be further simplified as follows:
\begin{align*}
u_P(\bar x)	&= m\int\D x'_1\, u_Q\bigl(\bar x+x'_1\bigr)\cdots \int\D x'_{m-2}\, u_Q\bigl(\bar x+x'_{m-2}\bigr) \\
	&\qquad \times
\int\D x'_{m-1}\, u_Q\bigl(\bar x+x'_{m-1}\bigr)u_Q\bigl((\bar x - x'_1-\cdots-x'_{m-2}) - x'_{m-1})\bigr) \\
	&= m\int\D x'_1\, u_Q\bigl(\bar x+x'_1\bigr)\cdots\int\D x'_{m-2}\, u_Q\bigl(\bar x+x'_{m-2}\bigr) \\
	&\qquad \times (u_Q*u_Q)\bigl((2\bar x - x'_1-\cdots-x'_{m-3})-x'_{m-2}\bigr) \\
	&\ \, \vdots\\
	&= m\bigl( \underbrace{u_Q*\cdots*u_Q}_{m\text{ times}} \bigr)(m\bar x).
\end{align*}

Now, let $P_k$ be the lattice resulting from dividing space-time into hyper-cubes of side $1/2^k$, and consider measures $\mu_k$ on $X_k=X_{P_k}$ given by
\[
\mu_k = \prod_{p\in P_k} \bigl( u_k\circ p^* \bigr)\D p^*.
\]
Our calculations above show that the compatibility conditions for the family $\{\mu_k\}$ read
\begin{equation} \label{ren pot recurrence}
\hat u_k(\xi) = \hat u_{k+1}(\xi/m)^m,\quad m=2^{d+1},
\end{equation}
where $\hat u = \mc F(u)$ is the Fourier transform of $u$.
\begin{ex}
For any $\beta\in[1,2]$, we get a family of compatible effective potentials by considering
\[
\hat u_k(\xi) = \e^{-\alpha_k\abs\xi^\beta},\quad \alpha_{k} = m^{1-\beta}\alpha_{k+1}\lo,
\]
which obviously solves~\eqref{ren pot recurrence}. Note that for $\beta>2$ we run into the problem that $\hat u_k$ is not positive definite. 
For $\beta=2$ we get a Gaussian measure, with
\[
V_k(x) = \frac{1}{4\alpha}\sum_i\frac{1}{2^{k(d+1)}} x_i^2 = \frac{1}{4\alpha}\sum_i \abs{p_i}x_i^2\Lo,\quad \alpha=\alpha_0\lo,
\]
so that this cylinder measure could be formally written
\[
\mu(\D x) = C\e^{-\lambda\int_S x^2}\D x,\quad \lambda = 1/4\alpha.
\]
For $\beta=1$ all of the $\alpha_k$'s are equal, 
and the resulting potential is given by
\[
V_k(x) = -\sum_i \log\left( \frac{1}{\pi}\frac{\alpha}{\alpha^2+x_i^2} \right) = \sum_i\log\left( 1+(x_i/\alpha)^2 \right) + C.
\]
This time, the sum does not converge as $k\rightarrow\infty$. In order to obtain an expression analogous to that for the Gaussian above, we introduce the scale-dependent coupling $\lambda(k) = 2^{k(d+1)}$, so that we can formally write
\[
\mu(\D x) = C\e^{-\lambda\int_S\log\left(1+(x/\alpha)^2\right)}\D x.
\]
Thus, we see that the $\beta=1$ case gives a cylinder measure which is a strong coupling limit for the potential $V(x) = \int_S\log\left(1+(x/\alpha)^2\right)$.
\end{ex}

\subsection{The $\phi^4$ field in arbitrary dimension}

Let us try to find a solution of~\eqref{ren pot recurrence} producing a measure corresponding to the formal expression
\begin{equation} \label{phi^4 measure}
\mu(\D x) = C\e^{-\lambda\int_S x^4}\D x.
\end{equation}
We start by considering effective solutions: given $k_0\in\N$, we ask that
\begin{equation} \label{effective phi^4}
u_{k_0}(x_i;\lambda) = \exp\left(-\frac{\lambda x_i^4}{2^{k_0(d+1)}}\right).
\end{equation}
It will be convenient to work in terms of the cumulants of $u_k(x;\lambda)$, i.e.\ the coefficients in the power series expansion
\[
\log\hat u_k(\xi;\lambda) = \sum_n c_n(k;\lambda)\xi^n,
\]
for then the renormalization flow simply reads
\[
c_n(k+1;\lambda) = m^{n-1}c_n(k;\lambda) = 2^{(d+1)(n-1)}c_n(k;\lambda).
\]
Let $\Set{c_n(\lambda)}$ be the set of cumulants of $u(x;\lambda) = \exp(-\lambda x^4)$, which are some numbers that can be computed recursively in terms of the moments
\[
\int x^{2n}\e^{-x^4}\D x = \frac{(-1)^n}{2} \Gamma\left(\frac{n}{2}+\frac{1}{4}\right),
\]
but whose precise value does not concern us here---besides the fact that infinitely many of them are non-zero. Since $u(x;\lambda) = u\left(\lambda^{1/4}x;1\right)$, one has that
\[
\log\hat u(\xi;\lambda) = \log\left( \lambda^{-1/4}\hat u\left(\lambda^{-1/4}\xi;1\right) \right),
\]
and therefore $c_n(\lambda) = \lambda^{-n/4}c_n(1)$ (except for $c_0$\lo, which is irrelevant, anyways). Now, let $\Set{c_n^{(k_0)}(k) | n\in\N}$ be the set of cumulants of the density $u_k$ obtained by renormalization from $u_{k_0}(x) = \exp\left(-\frac{\lambda x^4}{2^{k_0(d+1)}}\right)$. By the above calculations,
\(
c_n^{(k_0)}(k_0;\lambda) = \lambda^{-n/4}c_n(1),
\)
and therefore
\[
c_n^{(k_0)}(k;\lambda) = 2^{(d+1)(n-1)(k-k_0)}\lambda^{-n/4}c_n(1).
\]
Thus, 
we see that 
\[
\lim_{k_0\rightarrow\infty} c_n^{(k_0)}(k;\lambda) = 0.
\]
In other words, there is no cylinder measure such that~\eqref{phi^4 measure} holds. This will still be the case if we let $\lambda$ depend on $k_0$\lo, for if we want $c_n^{(k_0)}(k;\lambda)$ to stay finite as $k_0\rightarrow\infty$, we need
\[
\lambda(k_0) = O\left(2^{-k_0(d+1)(n-1)4/n}\right),
\]
whose dependence in $n$ makes it impossible to get all the coefficients right in the limit. One can check, indeed, that the best one can do is take
\[
\lambda(k_0)=O\left(2^{-2k_0(d+1)}\right)
\]
so that at least all the limits exist, but all of them are zero except for
\[
\lim_{k_0\rightarrow\infty}c_2^{(k_0)}(k;\lambda)
\]
and the resulting measure is Gaussian.

Let us summarize our findings. In arbitrary dimension, we have a formal candidate, given by our calculations above together with the perturbation expansion on the cutoff scale, of a local theory which is effectively $\phi^4$ at any given, fixed scale. Let $\nu^{(k_0)}(\D x)$ be the measure (with both free and interaction parts) of that field, where $k_0\in\N$ is the scale at which it is effectively $\phi^4$. We expect that 
\(
\lim_{k_0\rightarrow\infty}\nu^{(k_0)}(\D x)
\)
does not exist, because otherwise we would get in conflict with the non-existence of the corresponding cylinder perturbation of the free measure. This is irrespective of the dimension, because we have not put the interaction term in Wick order.

\appendix
\section*{Appendix: Gaussian measures}

A (centered) Gaussian measure on a finite dimensional space $X=\R^n$ is one of the form
\begin{equation} \label{Gaussian density}
\mu(\D x) = C\e^{-\frac12\langle \Gamma^{-1}x,x\rangle}\D x
\end{equation}
for an invertible $\Gamma:X^*\rightarrow X$. Since 
\(
\int_X \langle\xi, x\rangle\langle\eta, x\rangle\mu(\D x) = \langle \xi,\Gamma \eta\rangle,
\)
$\Gamma$ is called the \emph{covariance operator}. Expression~\eqref{Gaussian density} still makes sense when $X$ is an infinite dimensional Banach space, as a cylinder measure. Indeed, one can formally compute the characteristic function 
\[
\hat\mu(\xi)=\int\e^{-\I\langle\xi, x\rangle}\mu(\D x) = \e^{-\frac12\langle \xi,\Gamma \xi\rangle},
\]
and we conclude, given $e=(e_1\lo,\dots,e_n):X\rightarrow\R^n$, that on the finite dimensional quotient $X/\ker(e)$ we must put the measure 
\[
\mu_e(\D x_1\cdots\D x_n) = \int \frac{\e^{\I\sum\xi_i x_i}\prod \D\xi_i}{(2\pi)^{n/2}} \exp\left(-\frac12\sum \xi_i\xi_j\Langle e_i\lo,\Gamma e_j\Rangle\right),\quad x_i=e_i(x).
\]
The compatibility conditions are, then, easily verified.

The existence of a Radon extension of a cylinder measure on a reflexive Banach space is related, via Chebyshev, Prokhorov and Phillips' theorems, to the integrability of a coercive function of the norm. Let us review two well-known instances of this fact.
If $X$ is a Hilbert space and $\{e_n\}$ an orthonornal basis of eigenvectors of $\Gamma$,
one formally has
\[
\int_X \norm x^2\mu(\D x) =\tr\Gamma,
\]
and the following result holds---see~\cite{da2006introduction}, for instance, for a proof.
\begin{thm}
Let $\Gamma$ be a trace-class operator on the Hilbert space $X$. Then, there exists a unique Gaussian Radon measure on $X$ with covariance $\Gamma$.
\end{thm}
In the case of $L^p$ spaces, we have an analogous result, due to Rajput~\cite{rajput1972gaussian}. As an application, it can be readily seen that the Gaussian measure describing a free field, whose covariance kernel is the Green function of the differential operator $-\Delta+m^2$, does not radonify on $L^p(S)$.
\begin{thm} \label{rajput}
Let $S$ be a $\sigma$-finite measure space and $\gamma:S\times S\rightarrow\R$ a symmetric, positive definite, measurable fuction such that
\[
\int_S \gamma(s,s)^{p/2}\D s <\infty,
\]
where $1\leq p<\infty$. Then, the associated integral operator $\Gamma$ is well-defined as a bounded application $L^p(S)^*\rightarrow L^p(S)$ and is the covariance of a (unique) Gaussian Radon measure $\mu$ on $L^p(S)$. Conversely, the kernel $\gamma$ of the covariance of a Gaussian measure on $L^p(S)$ satisfies $\int_S\gamma(s,s)^{p/2}\D s <\infty$.
\end{thm}
\begin{rk}
The references given above are actually concerned with Borel measures, but  a Gaussian Borel measure on a Banach space which is either separable or reflexive is immediately Radon~\cite{sato1969gaussian}. 
\end{rk}

Now, consider the following instance of the reciprocal affirmation: given a Gaussian Radon measure $\mu$ on a Banach space $X$, are the powers of the norm integrable?
It turns out that there is  a bound on
\(
\int_X\norm x^p\mu(\D x)
\)
in terms of the variance 
\[
\sigma = \sup\Set{ \left(\int_X \abs{\langle\xi,x\rangle}^2\mu(\D x)\right)^{1/2} | \norm \xi_{X^*}\leq 1 },
\]
which is always finite.
Its existence is a consequence of a refinement due to Talagrand of Fernique's Theorem, as exposed in~\cite{ledoux1996isoperimetry}, for instance.
\begin{prop}
For every $\eps>0$, there exists a constant $C=C(\eps,p)$ such that
\[
\int_X \norm x^p\mu(\D x)\leq \eps+ C\sigma^p,
\]
whenever $\mu$ is a centered Gaussian measure on a separable Banach space $X$.
\end{prop}
\begin{proof}
For all of our unproved claims here see~\cite{ledoux1996isoperimetry}.
The forementioned Talagrand result provides the existence of an $r_0$ such that
\[
\mu\{\norm\cdot \geq \eps+\sigma r\} \leq \e^{-r^2/2 + \eps r},\quad r\geq r_0\lo.
\]
Take an $r_1\geq r_0$ such that 
\(
\e^{-r_1^2/2+\eps r_1}\leq 1/2.
\)
One has that 
\(
\mu\{\norm\cdot\leq\eps+\sigma r_1\}\geq 1/2
\)
and therefore 
\[
\int_X\norm x\mu(\D x)\leq \eps + \sigma r_1 + r_2
\]
whenever $r_2$ is such that $\e^{-r_2^2/2\sigma^2}<1/2$. Thus, for the $p=1$ case it suffices to take $r_2=c_0\sigma$ with $c_0>\sqrt{2\log 2}$.
The claim for arbitrary $p$ follows from the fact that norm moments are all equivalent for Gaussian measures; in particular, there exists a constant $C_p$ such that
\[
\left(\int_X\norm x^p\mu(\D x)\right)^{1/p} \leq C_p\int_X\norm x\mu(\D x). \qedhere
\]
\end{proof}
\begin{coro} \label{characterization of radon gaussians}
A Gaussian cylinder measure $\mu$ on a separable, reflexive Banach space $X$ admits a Radon extension if, and only if, \(
\int_X\norm x^p\mu(\D x) < \infty.
\)
\end{coro}
\begin{proof}
It only remains to check sufficiency. Using Chebyshev's inequality, one concludes that the measure of the balls of $X$, which are $*$-weakly compact, converge to 1 as their radii increases. Thus, by Prokhorov's theorem~\cite{prokhorov1956convergence, fremlin1974topological}, $\mu$ is Radon for the $*$-weak topology. We conclude by applying Phillips' theorem~\cite{schwartz1973radon, winkler1984note}.
\end{proof}

\bibliographystyle{amsplain}
\bibliography{/storage/emulated/0/eratosthenes/mathphys}

\end{document}